\documentclass[12pt]{article}
\usepackage{amsmath,amssymb}%,psfig,mathrsfs}
\usepackage{amsthm} %theorem environment option
\usepackage{latexsym}
\usepackage{mathtools}
\usepackage{tikz}
\usetikzlibrary{arrows}
\usetikzlibrary{intersections}
\usepackage{tkz-euclide}
\usepackage{hyperref}
\newtheorem{thm}{Theorem}[section]
\newtheorem{prp}[thm]{Proposition}
\newtheorem{lem}[thm]{Lemma}
\newtheorem{cor}[thm]{Corollary}

\begin{document}
\title{\bf Connected max cut is polynomial for graphs without $K_5\backslash e$ as a minor}
\date{}
\maketitle
\begin{center}
\author{
{\bf Brahim Chaourar} \\
{ Department of Mathematics and Statistics,\\Al Imam Mohammad Ibn Saud Islamic University (IMSIU) \\P.O. Box
90950, Riyadh 11623,  Saudi Arabia }\\{Correspondence address: P. O. Box 287574, Riyadh 11323, Saudi Arabia}}
\end{center}

\begin{abstract}
Given a graph $G=(V, E)$, a connected cut $\delta (U)$ is the set of edges of E linking all vertices of U to all vertices of $V\backslash U$ such that the induced subgraphs $G[U]$ and $G[V\backslash U]$ are connected. Given a positive weight function $w$ defined on $E$, the connected maximum cut problem (CMAX CUT) is to find a connected cut $\Omega$ such that $w(\Omega)$ is maximum among all connected cuts. CMAX CUT is NP-hard even for planar graphs. In this paper, we prove that CMAX CUT is polynomial for graphs without $K_5\backslash e$ as a minor. We deduce a quadratic time algorithm for the minimum cut problem in the same class of graphs without computing the maximum flow.
\end{abstract}

{\bf2010 Mathematics Subject Classification:} 90C27.
\newline {\bf Key words and phrases:} max cut; connected max cut; polynomial algorithm; min cut; graphs without $K_5\backslash e$ as a minor.
%%%%%%%%%%%%%%%%%%%%%%%%%%%%%%%%%%%%%%%%%%%%%%%%%%%%%%%%%%%%%%%%%%%%%%%%%%%%%%%%%%%%%%%%%%%%%%%%%%%
%%%%%%%%%%%%%%%%%%%%%%%%%%%%%%%%%%%%%%%%%%%%%%%%%%%%%%%%%%%%%%%%%%%%%%%%%%%%%%%%%%%%%%%%%%%%%%%%%%%
\section{Introduction}
%%%%%%%%%%%%%%%%%%%%%%%%%%%%%%%%%%%%%%%%%%%%%%%%%%%%%%%%%%%%%%%%%%%%%%

We refer to Bondy and Murty \cite{Bondy and Murty 2008} about graph theory terminolgy and facts.
\newline Given an undirected graph $G = (V, E)$ and positive weights $w_{ij} = w_{ji}$ on the edges $(i, j)\in E$, the maximum (respectively, minimum) cut problem (MAX CUT, (respectively, MIN CUT)) is that of finding the set of vertices $S$ that maximizes (respectively, minimzes) the weight of the edges in the cut $(S, V\backslash S)$ or $\delta (S)$ or $\delta (V\backslash S)$; that is, the weight of the edges linking all vertices of $S$ to those of $V\backslash S$. The (decision variant of the) MAX CUT is one of the Karp's original NP-complete problems \cite{Karp 1972}, and has long been known to be NP-complete even if the problem is unweighted; that is, if $w_{ij} = 1$ for all $(i, j)\in E$ \cite{Garey et al. 1976}. This motivates the research to solve MAX CUT in special classes of graphs. MAX CUT problem is solvable in polynomial time for the following special classes of graphs: planar graphs \cite{Barahona 1990, Hadlock 1975, Orlova and Dorfman 1972}, line graphs \cite{Guruswami 1999}, graphs with bounded treewidth, or cographs \cite{Bodlaender and Jansen 2000}. But the problem remains NP-complete for chordal graphs, undirected path graphs, split graphs, tripartite graphs, graphs that are the complement of a bipartite graph \cite{Bodlaender and Jansen 2000} and planar graphs if the weights are of arbitrary sign \cite{Terebenkov 1991}. Besides its theoretical importance, MAX CUT problem has applications in circuit layout design and statistical physics \cite{Barahona et al. 1988}. For a comprehensive survey of MAX CUT, the reader is referred to Poljak and Tuza \cite{Poljak and Tuza 1995} and Ben-Ameur et al. \cite{Ben-Ameur et al. 2014}. The best known algorithm for MAX CUT in planar graphs has running time complexity $O(n^{3/2} log n)$, where $n$ is the number of vertices of the given graph \cite{Shih et al. 1990}. The main result of this paper is to exhibit a quadratic time algorithm for a special variant of MAX CUT in graphs without the excluded minor $K_5\backslash e$.
\newline Let us give some definitions. Given an undirected graph $G = (V, E)$ and a subset of vertices $U$, a connected cut $\delta (U)$ is a cut where both induced subgraphs $G[U]$ and $G[V\backslash U]$ are connected. Special connected cuts are trivial cuts, i.e., cuts with one single vertex in one side (when this vertex is not a disconnecting vertex). The corresponding weighted variant of MAX CUT for connected cuts is called connected maximum cut problem (CMAX CUT). It is clear that MAX CUT and CMAX CUT are the same problem for complete graphs. Since MAX CUT is NP-hard for complete graphs (see \cite{Karp 1972}) then CMAX CUT is also NP-hard in the general case. Another theoretical motivation is that CMAX CUT gives a lower bound for MAX CUT.
\newline CMAX CUT has been proved NP-hard for planar graphs \cite{Haglin and Venkatesan 1991} and a linear time algorithm for series parallel graphs is presented in \cite{Chaourar 2017}. Some applications of CMAX CUT are: computing a market splitting for electricity markets \cite{Grimm et al. 2019, Kleinert and Schmidt 2018}, forest planning problems \cite{Carvajal et al. 2013}, phylogenetics \cite{Liers 2016}, image segmentation \cite{Vicente et al. 2008}, and graph coloring \cite{Hojny and Pfetsch 2018}.
\newline Let $G_1$ and $G_2$ be two graphs with $v_j$ a vertex (respectively, $e_j$ an edge) of $G_j, j = 1, 2$. The 1-sum (respectively, 2-sum) of $G_1$ and $G_2$ based on the vertices $v_j\in V(G_j)$ (respectively, edges $e_j\in E(G_j)$), $j=1, 2$, denoted $G_1\oplus_v G_2$ or $G_1\oplus G_2$ (respectively, $G_1\oplus_e G_2$ or $G_1\oplus_2 G_2$), is the graph obtained by identifying $v_1$ and $v_2$ (respectively, $e_1$ and $e_2$) on a new vertex $v$ (respectively, edge $e$), and keeping $G_j$ (respectively, $G_j\backslash e_j$), $j = 1, 2$, as they are. Moreover, we can define the 2-sum for two subsets $F_j\subseteq E(G_j)$, $j=1, 2$, as the edge set of the 2-sum of their corresponding subgraphs $(V(F_j), F_j)$, $j=1, 2$. Finally, for two classes $\mathcal X_j\subseteq 2^{E(G_j)}$, $j=1, 2$, $\mathcal X_1\oplus_2 \mathcal X_2=\{X_1\oplus_2 X_2$ such that $X_j\in \mathcal X_j$, $j=1, 2\}$.
\newline Let $\mathcal G_0$ be the class of wheels $W_n$ (where $n=|V(W_n)|\geq 4$), the prism $P_6$, $K_3$, and $K_{3, 3}$, and $\mathcal G(K_5\backslash e)$ be the class of graphs without $K_5\backslash e$ as a minor. In this paper, we prove that CMAX CUT is polynomial (time) for this class of graphs. For the best of our knowledge, this is the largest known class of graphs for which CMAX CUT is polynomial. We have the following characterization of $\mathcal G(K_5\backslash e)$ \cite{Wagner 1960}.
\begin{thm}
A graph $G\in \mathcal G(K_5\backslash e)$ if and only if $G$ is obtained by taking 1-sums and/or 2-sums of graphs of $\mathcal G_0$.
\end{thm}

Given a positive rational $\alpha$ and a class of graphs $\mathcal G$, we say that $\mathcal G$ is $\alpha$-polynomial for CMAX CUT (respectively, MIN CUT) if there exists a polynomial algorithm with running time complexity $O(n^{\alpha})$ which solves the considered problem for any graph $G\in \mathcal G$, where $n=|V(G)|$. In this case, we say that such a graph $G$ is $\alpha$-polynomial for the considered problem. The class of all connected cuts of a given graph $G$ is denoted by $\mathcal C(G)$. Moreover, for $e\in E(G)$, the class of connected cuts of $G$ containing $e$ is denoted by $\mathcal C_e(G)$.
\newline We can see the hardness of CMAX CUT by enumeration through the following.
\begin{prp}
\begin{align*}
|\mathcal C(K_n)|=\begin{cases}
                2^{n-1}-1 & if\> n\> is\> odd \\
            2^{n-1}-1-\frac{1}{2}{n\choose \frac{n}{2}} & if\> n\> is\> even \\
        \end{cases}
\end{align*}
\end{prp}

The remaining of the paper is organized as follows: in section 2, we prove that the class of graphs without the excluded minor $K_5\backslash e$ is 2-polynomial for CMAX CUT and MIN CUT without computing the maximum flow for the latter problem, and we conclude in section 3.

%%%%%%%%%%%%%%%%%%%%%%%%%%%%%%%%%%%%%%%%%%%%%%%%%%%%%%%%%%%%%%%%%%%%%%
\section{$\mathcal G(K_5\backslash e)$ is 2-polynomial for CMAX CUT and MIN CUT}
%%%%%%%%%%%%%%%%%%%%%%%%%%%%%%%%%%%%%%%%%%%%%%%%%%%%%%%%%%%%%%%%%%%%%%

First, we state the following lemma about the class of connected cuts when taking 2-sums.
\begin{lem} \
\begin{enumerate}
\item $\mathcal C(G_1\oplus G_2)=\mathcal C(G_1)\cup\mathcal C(G_2)$.
\item $\mathcal C(G_1\oplus_2 G_2)=\mathcal C(G_1/e_1)\cup\mathcal C(G_2/e_2)\cup [\mathcal C_{e_1}(G_1)\oplus_2 \mathcal C_{e_2}(G_2)]$.
\end{enumerate}
\end{lem}
\begin{proof}
(1) is trivial and (2) is direct because $\mathcal C(G_1\oplus_2 G_2)=\{ \Omega_j\in \mathcal C(G_j) : e_j\notin \Omega_j, j=1, 2\}\cup \{ \Omega_1\oplus_2 \Omega_2 : \Omega_j\in \mathcal C(G_j)$ and $e_j\in \Omega_j, j = 1, 2\}$.
\end{proof}

Now we start the process for proving the main result.
\begin{lem}
Let $\alpha>0$ be a rational. Then:
\begin{enumerate}
\item $G_1\oplus G_2$ is $\alpha$-polynomial for CMAX CUT if and only if $G_j$, $j=1, 2$ are too.
\item $G_1\oplus_2 G_2$ is $\alpha$-polynomial for CMAX CUT if and only if $G_j$, $j=1, 2$ are too.
\end{enumerate}
\end{lem}
\begin{proof}
It is not difficult to see that $\alpha$-polynomiality for CMAX CUT is preserved by minors. So if $G_1\oplus G_2$ (respectively, $G_1\oplus_2 G_2$) is $\alpha$-polynomial for CMAX CUT then $G_j$, $j=1, 2$ are too. Now for the inverse way, we will see the two cases separately.
\newline (1) According to the previous lemma, we need to solve two CMAX CUT problems, one in each $G_j$, $j=1, 2$. So the whole running time complexity for solving CMAX CUT in $G_1\oplus G_2$ is: $O(n_1^{\alpha}+n_2^{\alpha})=O((n_1+n_2)^{\alpha})=O(|V(G_1\oplus G_2)|^{\alpha})$, and we are done.
\newline (2) Let $w\in \mathbb{R}_+^{E(G_1\oplus_2 G_2)}$, $n_j=|V(G_j)|$, and $\Omega_j$ be a connected $w$-maximum cut containing $e_j$ in $G_j$, $j=1, 2$, among all connected cuts having the same property. It is not difficult to see that $\Omega_1\oplus_2 \Omega_2$ is a connected $w$-maximum cut in $G_1\oplus_2 G_2$. According to the previous lemma, we need to solve four CMAX CUT problems: in both $G_j/e_j$, $j=1, 2$, and in both $G_j$, $j=1, 2$, by changing the weight of $e_j$ to sum of all edges weights in order to force the corresponding solutions to contain this edge. Thus there exists an algorithm with running time complexity $O((n_1-1)^{\alpha}+(n_2-1)^{\alpha}+n_1^{\alpha}+n_2^{\alpha}+3)=O((n_1+n_2-2)^{\alpha})=O(|V(G_1\oplus_2 G_2)|^{\alpha})$ to solve CMAX CUT in $G_1\oplus_2 G_2$, and we are done.
\end{proof}

\begin{thm}
$\mathcal G_0$ is 2-polynomial for CMAX CUT.
\end{thm}
\begin{proof}
Let $G\in \mathcal G_0$ and $n=|V(G)|$. It suffices to prove that $|\mathcal C(G)|\leq n^2$ because, in this case, we have to find the maximum weighted element from at most $n^2$ elements.
\newline {\bf Case 1:} If $G$ is $K_3$ then $|\mathcal C(G)|=3\leq 9=n^2$.
\newline {\bf Case 2:} If $G$ is $P_6$ then $|\mathcal C(G)|=6+9+1=16\leq 36=n^2$.
\newline {\bf Case 3:} If $G$ is $K_{3, 3}$ $|\mathcal C(G)|=3(3+3+1)+3=24\leq 36=n^2$..
\newline {\bf Case 4:} Suppose now that $G$ is $W_n$. We will prove that $|\mathcal C(G)|=1+(n-2)(n-1)\leq n^2$. Let $C_{n-1}$ be the outside cycle of $W_n$, $\mathcal P_q$ be the class of simple paths of $C_{n-1}$ with $q$ vertices and $1\leq q\leq n-1$, and $\mathcal P=\bigcup_{q=1}^{n-1} \mathcal P_q$. Now let $\varphi$ be the application defined from $\mathcal P$ to $\mathcal C(G)$ such that $\varphi (P)=\delta (V(P))$. It is not difficult to see that $\varphi$ is a bijection. In the other hand, $|\mathcal P_{n-1}|=1$ and $|\mathcal P_q|=n-1$ if $1\leq q\leq n-2$. Thus $|\mathcal C(G)|=|\mathcal P|=1+(n-2)(n-1)$, and we are done.
\end{proof}

Now we can state our main result.
\begin{cor}
$\mathcal G(K_5\backslash e)$ is 2-polynomial for CMAX CUT.
\end{cor}
\begin{proof}
Direct from Theorem 1.1, Lemma 2.2, and Theorem 2.3.
\end{proof}

We have  a similar result for MIN CUT by using the following lemma \cite{Chaourar 2017}.
\begin{lem}
Given a connected graph $G=(V, E)$ and a positive weight function $w$ defined on $E$. Then any $w$-minimum cut is a connected cut of $G$.
\end{lem}

And we can state a version of Corollary 2.4 for MIN CUT.
\begin{cor}
  There exists a quadratic time algorithm for solving MIN CUT in $\mathcal G(K_5\backslash e)$ without computing the maximum flow.
\end{cor}
\begin{proof}
Direct from Lemma 2.5 and by adapting the quadratic algorithm of CMAX CUT.
\end{proof}

Note that, according to Proposition 1.2, $|\mathcal C(K_5)|=15\leq 25=n^2$. Thus we can get a larger class of graphs by taking 2-sums of $\mathcal G_0$ and copies of $K_5$ for which CMAX CUT and MIN CUT have quadratic running time complexity.
\newline In the other hand, by using Lemma 2.2 and similar decomposition theorems as for Theorem 1.1, we get linear time algorithms for CMAX CUT and MIN CUT (without computing the maximum flow) in large classes of graphs.
\newline Finally, we can have a quadratic running time complexity for the famous Hamitonian Cycle Problem (HC) in the following class of graphs.
\begin{cor}
  HC has quadratic running time complexity in graphs without the two excluded minors $K_5\backslash e$ and $K_{3, 3}$.
\end{cor}
\begin{proof}
Since the considered class of graphs is a subclass of both planar graphs and $\mathcal G(K_5\backslash e)$, then deciding if a given graph $G$ from this class contains a Hamiltonian cycle is equivalent to decide if a connected maximum cardinality cut (i.e., CMAX CUT with weights $w(e)=1$ for any edge $e$) of the dual graph $G^*$ has cardinality $n=|V(G)|$.
\end{proof}

This result is interesting because HC is NP-complete in maximal planar graphs \cite{Nishizeki et al. 1983}.

%%%%%%%%%%%%%%%%%%%%%%%%%%%%%%%%%%%%%%%%%%%%%%%%%%%%%%%%%%%%%%%%%%%%%%
\section{Conclusion}
%%%%%%%%%%%%%%%%%%%%%%%%%%%%%%%%%%%%%%%%%%%%%%%%%%%%%%%%%%%%%%%%%%%%%%
We have proved that CMAX CUT and MIN CUT have quadratic running time complexity for graphs with the excluded minor $K_5\backslash e$.
Further directions are improving this running time complexity and studying CMAX CUT in larger classes of graphs than $\mathcal G(K_5\backslash e)$.

%%%%%%%%%%%%%%%%%%%%%%%%%%%%%%%%%%%%%%%%%%%%%%%%%%%%%%%%%%%%%%%%%%%%%%
\iffalse
\noindent{ \bf Acknowledgements}

The author is grateful to anonymous referees in a previous version of this paper.
\fi

%%%%%%%%%%%%%%%%%%%%%%%%%%%%%%%%%%%%%%%%%%%%%%%%%%%%%%%%%%%%%%%%%%%%%%%%%%%%%%%%%%%

\end{document}